\documentclass{article}
\usepackage[utf8]{inputenc}
\usepackage[T1]{fontenc}
\usepackage[dvipsnames]{xcolor}
\usepackage{url}
\usepackage{booktabs}
\usepackage{amsfonts}
\usepackage{amsmath}
\usepackage{amssymb}
\usepackage{nicefrac}
\usepackage{tikz}
\usepackage{bm}
\usepackage{pgfplots}
\usepgfplotslibrary{patchplots}
\pgfplotsset{compat=1.18}
\usepackage{microtype}
\usepackage{tcolorbox}
\usepackage{enumitem}
\usepackage{graphicx}
\usepackage{subcaption}
\usepackage{algorithm}
\usepackage{algorithmic}
\usepackage[switch]{lineno}
\usepackage{dirtytalk}
\usepackage{mathtools}
\usepackage{amsthm}
\usepackage{thmtools}
\usepackage{thm-restate}
\usepackage{longtable}
\usepackage{authblk}
\usepackage{float}
\usepackage{natbib}
\usepackage[verbose=true,letterpaper]{geometry}
\usepackage[colorlinks = true, linkcolor=NavyBlue,citecolor=ForestGreen]{hyperref}
\usepackage[nameinlink]{cleveref}

\newtheorem{theorem}{Theorem}
\newtheorem{lemma}{Lemma}
\newtheorem{definition}{Definition}
\newtheorem{observation}{Observation}
\newtheorem{corollary}{Corollary}

\newcommand{\brset}{\mathrm{BR}}
\newcommand{\gm}{\text{A}}

\title{Do Not Discretize, Optimize: Almost Greedy Fictitious Play}

\author[1,2,3]{Evangelos Markakis}
\author[1,2]{Christodoulos Santorinaios}

\affil[1]{Athens University of Economics and Business, Greece}
\affil[2]{Archimedes, Athena Research Center, Greece}
\affil[3]{Input Output Group (IOG), Greece}

\affil[ ]{\texttt{\{markakis, santgchr\}@aueb.gr}}

\begin{document}
\maketitle
\begin{abstract}
    Our work revolves around Fictitious Play, one of the first iterative methods that is known to converge to a Nash equilibrium in zero-sum games. In recent years, there has been a revived interest, due to applications in various machine learning problems, which has motivated a line of work on its convergence properties and on proposing new variants of the initial algorithm. Our paper is along this direction and introduces one new variant, which we refer to as \textit{Almost Greedy Fictitious Play}. The proposed algorithm \textit{greedily} attempts to find the optimal stepsize at each iteration but its search space is constrained and includes \textit{almost} all the line between the cumulative mixed strategy and the current best response. Our main result is that the method achieves an instance dependent convergence rate of $\mathcal{O}(1/T)$ with respect to the duality gap. This matches the rate of Continuous Fictitious Play, and offers an alternative to discretization. We complement our theoretical findings with experiments that demonstrate the effectiveness of the method.
\end{abstract}

\section{Introduction}

Computing Nash equilibria in two-player bilinear zero-sum games is a cornerstone of Algorithmic Game Theory with a research trajectory spanning several decades. While zero-sum games are solvable via linear programming, such centralized approaches often prove computationally prohibitive in large-scale problems. This has led to a resurgence of interest in iterative methods suitable for machine learning applications such as formulating Generative Adversarial Networks (GANs), \cite{DBLP:conf/nips/GoodfellowPMXWOCB14}. In this context, the goal is to design simple iterative algorithms that converge efficiently to a Nash equilibrium, where no player has an incentive to deviate.

To this end, we revisit the classic Fictitious Play algorithm, \cite{brown1949some}. This is one of the first algorithms ever proposed for zero-sum games, based on a very intuitive idea. At every iteration each player selects a best response to the empirical average of the other player's history. The new strategy profile at each iteration is then updated to the new empirical average, where the probability of each pure strategy equals its selection frequency so far. Alternatively, we can think of the update rule as selecting a particular convex combination between the (mixed) strategy of the previous iteration and the pure best response strategy identified in the current iteration.

Motivated by the latter formulation of Fictitious Play, one can think of defining greedy variants as follows: if at iteration $t$, we are at the profile $(x_t, y_t)$ and the best response of the two players are $i, j$ respectively, we could choose in a greedy manner the step size $\eta_t$, which is the probability we assign to the best responses, while keeping a probability of $1-\eta_t$ for $(x_t, y_t)$. One natural metric for making such a choice is the Duality Gap, which represents the sum of regrets of the two players. Therefore, we could pick the stepsize so as to minimize the duality gap over all possible convex combinations, i.e., over $\eta_t\in [0,1]$. We refer to this variant as Greedy Fictitious Play.

In order to highlight some further intuition on this variant, we illustrate how Greedy Fictitious Play runs for the classic Rock-Paper-Scissors (RPS) game in \Cref{fig:agfp_rps}. We also show the stepsizes $\eta_t$ computed during the first 10 iterations in \Cref{tab:10steps}. 

\begin{figure}[H]
    \centering
    \includegraphics[scale=.7]{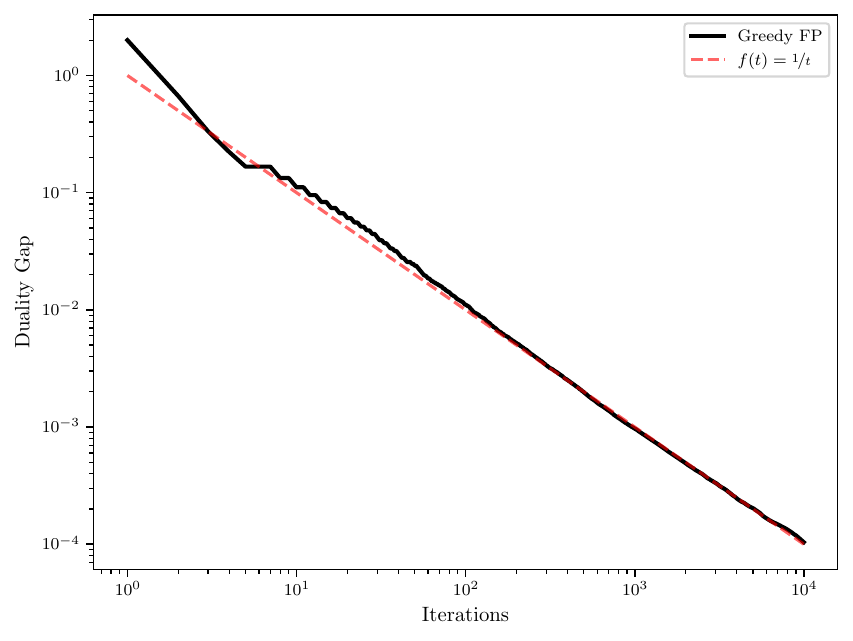}
    \caption{Greedy Fictitious Play in the Rock-Paper-Scissors game}
    \label{fig:agfp_rps}
\end{figure}

The encouraging news from this illustration is that the duality gap seems to be dropping at a rate of $1/t$ over time. Let us also take a closer look at the evolution of the step sizes. Under our initialization, both players start by playing Rock. Given this, at iteration 1, the theoretically optimal greedy choice is to set $\eta_1 = \nicefrac{2}{3}$. In this case, the next profiles would be $\nicefrac{1}{3}$ Rock and $\nicefrac{2}{3}$ Paper. Now the best response is not unique: it can be either Paper or Scissors. Assuming we are breaking the ties lexicographically, the best response would again be Paper for each player, resulting in $\eta_2 = 0$ and the algorithm cannot further progress.

However, due to numerical precision issues, the algorithm selects a slightly larger stepsize which leads to Scissors being selected as a unique best response, instead of Paper. This also occurs in the next few iterations. The next interesting thing to observe is at iteration 5. There, theoretically, the correct step size would have been $\eta_t = 0$, and the algorithm would get stuck. But, again due to the numerical precision, the computed step size is slightly positive, and this added noise eventually helps the method to get unstuck. 

\begin{table}[H]
\centering
\begin{tabular}{cc}
\toprule
$t$ & Greedy stepsize \\
\midrule
1  & 0.666666668 \\
2  & 0.500000004 \\
3  & 0.333333340 \\
4  & 0.250000004 \\
5  & $3.7 \times 10^{-9}$ \\
6  & $3.7 \times 10^{-9}$ \\
7  & 0.200000007 \\
8  & $3.7 \times 10^{-9}$ \\
9  & 0.166666675 \\
10 & $3.7 \times 10^{-9}$ \\
\bottomrule
\end{tabular}
\caption{Greedy stepsizes for the first 10 iterations of RPS}\label{tab:10steps}
\end{table}

This also occurs in subsequent iterations, as shown in \Cref{fig:agf_rps_steps}, and it motivates the question of whether we could introduce this added noise explicitly in the method.

\begin{figure}[H]
    \centering
    \includegraphics[scale=.7]{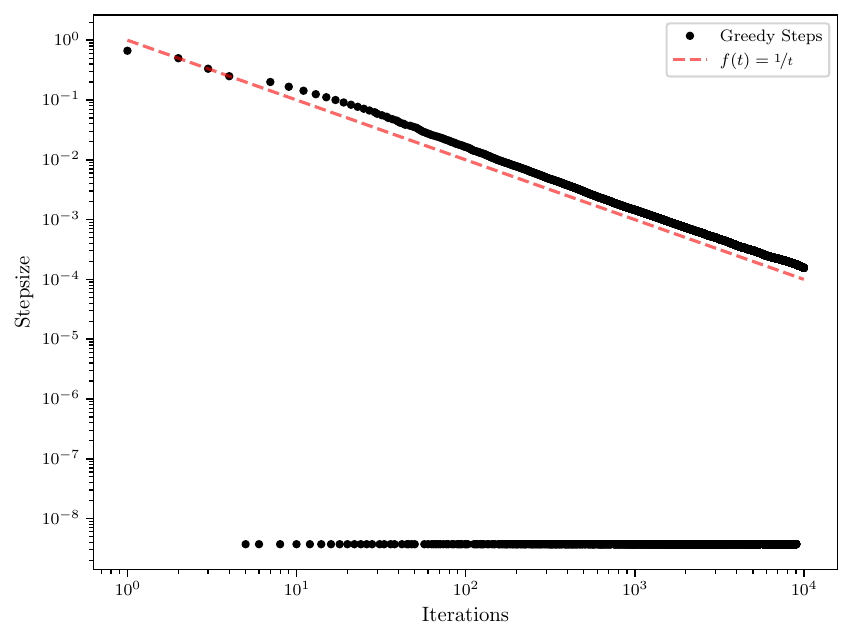}
    \caption{The steps the algorithm selected greedily}
    \label{fig:agf_rps_steps}
\end{figure}

\subsection{Contributions}

Motivated by the previous considerations, we study a slight variation of the algorithm discussed previously, which we term {\it Almost Greedy Fictitious Play} (AGFP). The only difference with Greedy Fictitious Play is that we do not optimize the selection of the stepsize $\eta_t$ over $[0,1]$, but instead over $[\delta,1]$, for some appropriately selected noise $\delta$.

We study this method both theoretically and experimentally. We first analyze the convergence rate of the method w.r.t. the standard duality gap metric. In particular, for a game with payoff matrix $\gm$, we prove that it converges to an approximate Nash equilibrium at a rate of $O(\frac{1}{\kappa_A T})$, where $\kappa_\gm$ is a condition-like number, formally defined in \Cref{def:cnum}. Such game-dependent quantities appear often in the convergence analysis of other algorithms (e.g. in the gradient descent family), where the general goal is to move from a $\mathcal{O}\big(\frac{1}{T}\big)$ rate to the improved $\mathcal{O}{(c^T)}$ (see for instance \cite{DBLP:conf/iclr/WeiLZL21}). In our case, this constant is necessary even for the sublinear rate.

We then turn to an experimental evaluation where our results convey an even better picture. We have tested our method in random games, and we first observe that it is a faster method than the classic FP algorithm. Furthermore, as demonstrated in \Cref{sec:exp}, the convergence rate seems to be in the order of $O(n/T)$, where $n$ is the dimension of the game (replaced by $\max\{n, m\}$ if the payoff matrix is $m\times n$). This matches the rate of Continuous Fictitious Play and offers an alternative to other related methods. Overall, we view our approach as a promising idea that motivates even further research on FP-type algorithms.    

\subsection{Related Work}

The literature on Fictitious Play is quite extensive and impossible to exhaustively list here. We discuss however below what we feel are the most relevant works.

As mentioned, the algorithm was originally proposed by \cite{brown1949some,brown1951iterative} and, in her seminal work, \cite{robinson1951iterative} proved the convergence of the dynamics for zero-sum games via an inductive argument. The proof implied by her rate was exponential in the dimensions of the payoff matrix. Shortly after, \cite{karlin2003mathematical} conjectured that the true rate is polynomial, and in fact $\mathcal{O}\big(\frac{1}{\sqrt{t}}\big)$. \cite{brandt2010rate} showed that exponential time may pass until the first time an equilibrium action is being played but this does not imply that the conjecture is false. \cite{daskalakis2014counter} disproved the conjecture by constructing an exponentially slow dynamic for the identity matrix, with their lower bound essentially matching Robinson's upper bound. Critically, their idea hinged on an adversarial tie-breaking rule. This left room to explore more natural tie-breaking rules, e.g. lexicographic, under which \cite{abernethy2021fast} proved that Karlin's conjecture holds, albeit only for diagonal matrices. Another class for which Karlin's conjecture is now verified is a generalization of the Rock-Paper-Scissors game, due to \cite{DBLP:conf/colt/LazarsfeldPSW25}. Very recently, \cite{wang2025tie} constructed a game for which Fictitious Play converges in $\mathcal{O}\big(t^{-1/3}\big)$.
Beyond the zero-sum world, a series of convergence or non-convergence results has been produced over the years, starting with \cite{miyasawa1961convergence} who proved convergence for $2\times 2$ games under a specific tie-breaking rule. The first non-convergence result is a $3 \times 3$ game due to \cite{shapley19641}. Three decades later, \cite{monderer1996a2} showed that without a correct tie-breaking even $2\times 2$ games may fail to exhibit convergence. A popular class of games were convergence was proved is potential games by \cite{monderer1996fictitious,monderer1996potential}, albeit under an exponential worst case rate, \cite{DBLP:conf/nips/PanageasPSC23,anagnostides2026doubly}.

Regarding variants of Fictitious Play, we should start our exposition with Continuous time Fictitious Play, which converges at a rate of $\mathcal{O}\big(\frac{1}{T}\big)$, as shown by \cite{harris1998rate}. The author claimed that, as a result, the discrete dynamics enjoy the same performance, but it turned out that discretizing the dynamics is non-trivial and slows down the algorithm (see the discussion by \cite{abernethy2021fast} as well). For recent progress of Continuous Fictitious Play see \cite{DBLP:journals/corr/OstrovskiS13} and the references therein. A second noteworthy variant is that of Stochastic Fictitious Play \citep{fudenberg1993learning}, which replaces the best response with a softmax function. The convergence of the method was established by \cite{hofbauer2002global}. We also point the reader to \cite{fudenberg1995consistency}, that defined the notion of consistency, in an attempt to explain how the selection of the best response affects the algorithm. For a survey of some of these results, the work of \cite{krishna1997learning} can be consulted. In the last couple of years, newly proposed variants include Anticipatory Fictitious Play \citep{DBLP:conf/ijcai/CloudWK23}, where each player best responds to the history of the game and the action that the opponent would have made in the next iteration of Fictitious Play, and Optimistic Fictitious Play \citep{lazarsfeld2025optimism}, that adapts the popular notion of Optimism into Fictitious Play. 

Regarding applications, Fictitious Play has been studied in the context of control theory: \cite{ma2017forecasting,marden2009joint} and in stochastic games: \cite{DBLP:conf/icml/BaudinL22,sayin2022fictitious,DBLP:conf/nips/BaudinL22}. It has also inspired algorithms in Reinforcement Learning, e.g. \cite{DBLP:conf/iclr/MullerORTPLHMLH20,DBLP:conf/icml/YangLLZZW18}. It has been adapted to extensive form games as Fictitious Self-Play \cite{heinrich2015fictitious}, which in turn was one of the components of the recent AI breakthrough that achieved Grandmaster level in the game of Starcraft \cite{vinyals2019grandmaster}. For a survey about the MARL setting we point to \cite{yang2020overview} and for a more recent work about Fictitious Play in multiplayer games to \cite{DBLP:conf/ijcai/0001DLM0Y022}.

We complete this section with a brief mention to the Frank-Wolfe method, \citep{frank1956algorithm}. There is an inherent connection between the algorithm and Fictitious Play, given that Fictitious Play can be seen as an instantiation of Frank-Wolfe with stepsize $\nicefrac{1}{(t+1)}$. In the same light, our proposed variant is connected with Frank-Wolfe with an optimal stepsize. We point the reader to \cite{DBLP:conf/aistats/GidelJL17} for more information, and stress that the lack of strong convexity that is present in the Frank-Wolfe analyses is precisely what complicates our proof in the sequel.

\section{Preliminaries}
\paragraph{Notation.} Let $[n] = \{1, \dots, n\}$, $\Delta_n = \{x \in \mathbb{R}^n: x_i \ge 0, \sum_{i=1}^n x_i = 1\}$ be the $n$-dimensional probability simplex and $e_i$ the $i$\textsuperscript{th} elementary basis vector of $\mathbb{R}^n$. A zero-sum game with $m$ (resp. $n$) pure strategies for the row (resp. column) player is described by a payoff matrix $\gm \in \mathbb{R}^{m\times n}$. Pure strategies correspond to basis vectors and mixed strategies to points in the players' respective simplices. Finally, a strategy profile is a point in $\Delta_m \times \Delta_n$.

\begin{definition}[Nash equilibrium, \cite{nash1951noncooperative}]
 A strategy profile $(x^*, y^*)$ is a Nash equilibrium of the zero-sum game $(\gm, -\gm )$ if and only if 
\[ 
(x^*)^\top \gm e_j \ge (x^*)^\top \gm y^*  \ge  e^\top_i \gm y^*\]  
for any $i \in [m]$ and $j \in [n]$.
\end{definition}
A Nash equilibrium is a profile where no player has a unilateral incentive to deviate. A common distance measure of a profile $(x,y)$ to a Nash equilibrium is the following.

\begin{definition}[Duality Gap]
    For a strategy profile $(x,y)$ the Duality Gap function is defined as
    \[\psi(x, y) = \max_{i \in [m]} e^\top_i \gm y - \min_{j \in [n]} x^\top \gm j.\]
\end{definition}
From Von Neumann's Duality theorem \citeyearpar{vonneumann1928theorie} it follows that the duality gap is nonnegative and, crucially, \[\psi(x,y) = 0 \iff (x,y) \text{ is a Nash equilibrium.} \]

We are also interested in approximate Nash equilibria, and we will use the standard version of approximation as defined below, assuming that the matrix entries are normalized in $[0,1]$.
\begin{definition}[$\epsilon$-Nash equilibrium]
A strategy profile $(x,y)$ is an $\epsilon$-Nash equilibrium (in short, $\epsilon$-NE), 
if and only if, for any $i\in [m], j\in [n],$ it holds that
\begin{equation*}
x^\top A y +\epsilon\geq e_i^\top A y, \text{ and, }
x^\top A y -\epsilon\leq x^\top A e_j.
\end{equation*}
\end{definition}

One can also easily see that when $\psi(x,y) \leq \epsilon$, then $(x,y)$  is an  $\epsilon$-Nash equilibrium.

\subsection{Fictitious Play}
Fictitious Play describes an iterative process where each player best responds to the empirical average of the other player's history. The final output is the averages of the actions. Alternatively, we can view Fictitious  Play as a dynamic that maintains a running average of the aforementioned best responses. We will present the second formalism, as it leads naturally to our approach. Let $(x_t, y_t)$ be the strategy profile produced at time $t$. Then at the next iteration, the update is performed as follows.
\begin{align*}
   &x_{t+1} = \frac{t}{t+1} x_t + \frac{1}{t + 1} e_{i_t},\quad i_t \in \arg\max_{i \in [m]} e_i^\top \gm y_t, \\
   &y_{t+1} = \frac{t}{t+1} y_t + \frac{1}{t + 1} e_{j_t},\quad j_t \in \arg\min_{j \in [n]} x_t^\top\gm e_j.
\end{align*}

\section{Almost Greedy Fictitious Play}

We begin by introducing our algorithm. The pseudocode can be seen as \Cref{alg:agfp}. The main idea behind our FP variant is that after we compute the pure best responses at a round $t$, each player plays a convex combination between the previous iterate and the best response (as in Fictitious Play). But the difference now is that the coefficients of the convex combination are selected by optimizing the duality gap function. In particular, the coefficient $\eta_t$ in the algorithm's pseudocode is the minimizer of the duality gap within the interval $[\delta, 1]$, where $\delta$ is an appropriately small but positive parameter (as implied by our analysis in the sequel) . A more greedy choice would be to minimize over the entire interval $[0, 1]$. However, we cannot guarantee a non-zero optimal step at each iteration. Hence, the need for some noise $\delta$, which can also be seen as a form of an implicit regularization.
Given this, we refer to our FP variant as Almost Greedy Fictitious Play.

\begin{algorithm}[tb]
    \caption{Almost-Greedy Fictitious Play}\label{alg:agfp}
    \textbf{Input}: Matrix $\gm$\\
    \textbf{Parameter}: Number of iterations $T$ and parameter $\delta$\\
    \textbf{Output}: Strategy profile $(x,y)$
    \begin{algorithmic}[1]
        \STATE{Pick an initial state $z_0 = (x_0, y_0)$.}
        \FOR{$t = 1, \dots, T$} 
        \STATE{Pick $(i,j) \in \brset(y_{t-1}) \times \brset(x_{t-1})$.}
            \STATE{$\eta_t = \arg\min\limits_{\eta\in[\delta,1]} \psi(z(\eta))$, where $z(\eta) = (1-\eta)z_{t-1} + \eta (e_i, e_j)$ }
        \STATE{$z_t = (1-\eta_t)z_{t-1} + \eta_t(e_i,e_j)$}
        \ENDFOR
        \STATE \textbf{return} $z_T = (x_T, y_T)$
    \end{algorithmic}
\end{algorithm}

\subsection{Convergence Analysis}

To prove that the algorithm eventually converges, first we must identify under which conditions the duality gap of the next iterate decreases. The following lemma highlights an important case where this happens.

\begin{lemma}[Decreasing step]\label{lem:decrease}
If $i_t$ (the best response of the row player against $y_t$) remains a best response against $y_{t+1}$ and respectively $j_t$ also remains a best response against $x_{t+1}$, the duality gap decreases. More specifically, it holds that
\[
\psi(x_{t+1}, y_{t+1}) = (1-\eta_t) \psi(x_t, y_t)
\]
where $\eta_t$ is the selected step at time $t$.
\end{lemma}
\begin{proof}
We calculate
\begin{align*}
\psi(x_{t+1}, y_{t+1}) - \psi(x_{t}, y_{t}) &=
\left(\max_i e_i^\top \gm y_{t+1} - \min_j x_{t+1}^\top \gm e_j\right) - \left(\max_i e_i^\top \gm y_{t} + \min_j x_{t}^\top \gm e_j\right) \\
&= e_{i_t}^\top \gm y_{t+1} - x_{t+1}^\top \gm e_{j_t} - e_{i_t}^\top \gm y_{t} + x_{t}^\top \gm e_{j_t} \\
&= e_{i_t}^\top \gm (y_{t+1} - y_t) - (x_{t+1} - x_t)^\top \gm e_{j_t} \\
&= \eta_t e_{i_t}^\top \gm (e_{j_t} - y_t) - \eta_t (e_{i_t} - x_t)^\top \gm e_{j_t}  \\
&= \eta_t (\gm_{i_tj_t} - e_{i_t}^\top \gm y_t - \gm_{i_tj_t} + x_t^\top \gm e_{j_t}) \\
&= -\eta_t (\max_i e_i^\top \gm y_t - \min_j x_t^\top \gm e_j ) \\
&= -\eta_t \psi(x_t, y_t).
\end{align*}
The proof is completed by rearranging the terms.
\end{proof}

The next lemma shows a type of updates under which a pure strategy remains a best response when going from iteration $t$ to $t+1$, if it was a unique best response at $t$.

\begin{lemma}[Stability of unique best responses]\label{lem:uniquestab}
 Let $j$ be a pure strategy of the column player. If $i^*$ is the unique best response against a strategy $y$, then there exists a $\bar{\theta} > 0$ such that $i^*$ is a best response against $y^\prime = (1-\theta)y + \theta e_j$ for any $\theta \in [0, \bar{\theta}]$. 
\end{lemma}
\begin{proof}
Fix $j$. We will argue by comparing the payoffs of the strategies $i^*$ and $i$ of the row player, against strategy $y^\prime = (1-\theta_i)y + \theta_i e_j$, for any $i\in [n]$. For clarity, we first carry out the analysis parametrically, using $\theta_i$ in the definition of $y^\prime$, when we compare with strategy $i$. We will later extract an upper bound for all the $\theta_i$'s.
\begin{align}
e_{i^*}^\top \gm  y^\prime - e_i^\top \gm y^\prime &=
(1-\theta_i) e_{i^*}^\top \gm y + \theta_i \gm_{i^*j}
 - [(1-\theta_i) e_i^\top \gm y + \theta_i \gm_{ij}] \nonumber \\ 
&= (1-\theta_i) (e_{i^*}^\top \gm y - e_i^\top \gm y) + \theta_i (\gm_{i^*j} - \gm_{ij}) \nonumber \implies \\
e_{i^*}^\top \gm  y^\prime - e_i^\top \gm y^\prime &= \theta_i\underbrace{[\gm_{i^*j} - \gm_{ij} - e_{i^*}^\top \gm y + e_i^\top \gm y]}_{a_{ij}}
+\underbrace{e_{i^*}^\top \gm y - e_i^\top \gm y}_{b_i} \tag{$\star$}\label{eq:star}
\end{align}
where $b_i>0$ regardless of $i$, since $i^*$ is the unique best response against $y$. Now, if $a_{ij}$ is also nonnegative for any $i$ we have established that the right-hand side expression is positive, hence $i^*$ is a better response than $i$ against $y^\prime$. In other words, when $a_{ij} \ge 0$, $i^*$ is a strictly better response than $i$ for any $\theta_i \in [0,1]$. Now, assume that for some $i$, we have $a_{ij} < 0$. Then $i^*$ is at least as good a response as $i$, when $a_{ij} \theta_i + b_i \geq 0$. This implies that it suffices to have $\theta_i$ upper bounded as follows.
\[
 \theta_i \leq \frac{e_{i^*}^\top \gm y - e_i^\top \gm y}{\gm_{ij} - \gm_{i^*j} - e_{i}^\top \gm y + e_{i^*}^\top \gm y} = \bar{\theta}_i(i^*).
\]
 
Note that the quantity $\bar{\theta}_i(i^*)$ is positive since we assumed that $a_{ij}<0$, and since $i^*$ is a unique best response. For completeness, when $a_{ij} \ge 0$, we define $\bar{\theta}_i(i^*) = 1$. To complete the proof, we set $\bar{\theta}(i^*) = \min_i \bar{\theta}_i(i^*)$. Then, $i^*$ remains a best response against $y^\prime$, for any $\theta \in [0, \bar{\theta}(i^*)]$, potentially tied with other strategies.
\end{proof}

The above lemma offers a crucial insight in identifying situations where \Cref{lem:decrease} can be exploited but it is too strict, as it concerns unique best responses. To relax it, we need the following definition, which captures what we will target when best responses will not be unique.

\begin{definition}[Linearly indistinguishable   responses]\label{def:lind}
Let $y_1, y_2$ be strategies of the column player. We say that two pure strategies $i_1, i_2$ of the row player are linearly indistinguishable (l. ind.) w.r.t. $y_1, y_2$, if 
    \[e_{i_1}^\top \gm y_1 = e_{i_2}^\top \gm y_1 \text{ and } e_{i_1}^\top \gm y_2 = e_{i_2}^\top \gm y_2.\]
\end{definition}

\Cref{def:lind} trivially implies that the strategies $i_1, i_2$ produce the same payoff against any strategy on the line between $y_1$ and $y_2$. The next observation essentially extends \Cref{lem:uniquestab} for the case of non-unique best responses and guides us for the choice of $\delta$ in \Cref{alg:agfp}. 

\begin{observation}[Stability of l.ind. best responses]\label{obs:stab} Consider a (possibly mixed) strategy $y$ and a pure strategy $j$ of the column player. If there are multiple best responses against $y$ but they are linearly indistinguishable w.r.t $y$ and $j$, then there exist a $\bar{\theta} > 0$ such that they are best responses against any convex combination $y^\prime= (1-\theta)y + \theta e_j$ for $\theta \in [0, \bar{\theta}]$. 
\end{observation}

\begin{proof}
Let $i_1^*, i_2^*$ be two l. ind. best responses w.r.t. $y$ and $j$. We follow the same reasoning as in the proof of \Cref{lem:uniquestab}, and compare $e_{i_1^*}^\top \gm  y^\prime$ against $e_{i}^\top \gm  y^\prime$ for any other pure strategy $i$.  For $i\not\in BR(y)$, we can follow exactly the proof of \Cref{lem:uniquestab} and obtain a value $\bar{\theta}_i(i_1^*)$ so that $i_1^*$ is a better response than $i$ against the convex combination $y^\prime =  (1-\theta)y + \theta e_j$ when $\theta \in [0, \bar{\theta}_i(i_1^*)]$. Set now $\bar{\theta}(i_1^*) = \min_{i\not\in BR(y)} \bar{\theta}_i(i_1^*)$. Finally for $i\in BR(y)$, since $i$ and $i_1^*$ are linearly indistinguishable w.r.t. $y$ and $j$, observe that the terms $a_{ij}$ and $b_i$ within the proof of \Cref{lem:uniquestab} are 0, therefore all these strategies yield the same payoff against $y^\prime$. Hence, we can conclude that $i_1^*$ is a best response to $y'$ for any $\theta \in [0, \bar{\theta}(i_1^*)]$. Using the same arguments, there exists $\bar{\theta}(i_2^*)$ so that $i_2^*$ satisfies the same property. We can now take $\bar{\theta} = \min \{\bar{\theta}(i_1^*), \bar{\theta}(i_2^*) \}$, and we are done. 
\end{proof}

Next, we deal with the case where the optimal stepsize would have been in $[0, \delta)$, but instead we select $\delta$. We call such a stepsize a tie-breaking one, because its role is not to reduce the duality gap but merely to allow the algorithm to find a good direction where the best responses are not unique (or l. ind.).

\begin{lemma}[Tie-breaking step]\label{lem:tie}
There is a sufficiently small $\delta$ such that the following holds: Suppose we run \Cref{alg:agfp} using $\delta$, and at time $t$ it sets $\eta_t = \delta$ to produce $(x_t, y_t)$. Let $j$ be the best response to $y_{t-1}$ found at time $t$. Then at time $t+1$, if there are multiple best responses of the row player to $y_t$, then they are linearly indistinguishable w.r.t. $y_t$ and $j$. Similarly for the column player.
\end{lemma}

\begin{proof}
    Let $(x_{t-1},y_{t-1})$ be a pair of iterates where the minimizer over the line was actually less than $\delta$, and instead \Cref{alg:agfp} selected $\delta$ and let $(x_t, y_t)$ be the next pair, produced based on the choice of best responses $i,j$, so that $j\in BR(x_{t-1})$. This means that $y_t = (1-\delta)y_{t-1} + \delta e_j$. Let us focus on the row player. Assume that $i_1, i_2$ are best responses against $y_t$. We claim that there exists a sufficiently small $\delta>0$ such that $i_1, i_2$ are also best responses against $y_{t-1}$. This follows by arguments similar to the proof of \Cref{lem:uniquestab}.
    Suppose now that $i_1$ and $i_2$ are not l.ind. w.r.t.  $y_{t-1}$ and $j$. Note that since $i_1, i_2$ are best responses against $y$  and they are not l. ind., it must hold that $\gm_{i_1 j} - \gm_{i_2j}\neq 0$. If we set in \Cref{eq:star}  $y^\prime = y_t, i^* = i_1$, and $i=i_2$, we obtain
    \[
    e^\top_{i_1} \gm y_t - e^\top_{i_2} \gm y_t = \delta(\gm_{i_1 j} - \gm_{i_2j}) \ne 0.
    \]
    This is a contradiction, since $i_1, i_2$ are both best responses to $y_t$, hence, there is no pair of non l.ind. strategies in $\brset(y_t)$.\end{proof}

The issue now is that while the tie-breaking step guarantees us a drop direction in the next iteration, it may have increased the duality gap. The next lemma describes the trade-off and it essentially boils down to the fact that as long as the duality gap is large, the effect of $\delta$ will be negligible.

\begin{lemma}[Bounded increase after decrease]\label{lem:bndinc}
    If at time $t$ AGFP did not take a decreasing step, there exists a sufficiently small $\delta$ such that $\psi(z_{t+1}) < \psi(z_{t-1})$.
\end{lemma}
\begin{proof}
 If at time $t-1$ the algorithm did not make a decreasing step, it has performed a tie-breaking one. Then \Cref{lem:tie} and the choice of $\delta$ guarantees the decrease condition for time $t$, via either \Cref{lem:decrease} or \Cref{obs:stab}, depending on whether the outcome of the tie-breaking leads to a unique best response or multiple l. ind. responses.
 Hence, the step at time $t$ is decreasing, i.e. $\psi(z_t) = (1- \eta_{t-1}) \psi(z_{t-1})$. By the convexity of $\psi$ we have that 
    \[\psi(z_{t+1}) \le (1-\delta)\psi(z_t) + \delta \psi(e_{i_t},e_{j_t})  \le \psi(z_t) + 2\delta\]
    where we used the bound $\psi(e_{i_t},e_{j_t}) \le 2$. Combining the two relations and as long as $\delta \le \frac{\eta_{t-1}\psi(z_{t-1})}{4}$ we have
    \[
    \psi(z_{t+1}) \le \Big(1 - \frac{\eta_{t-1}}{2}\Big) \psi(z_{t-1})
    \]
    which completes the proof.\end{proof}

If we were minimizing a strongly-convex strongly-concave function, we would have had all the ingredients to obtain the claimed rate, by relating $\eta$ with $\psi$ via the strong convexity constants. Unfortunately, this is not the case for our problem. To make progress, we define a quantity that depends on the matrix $\gm$. 

\begin{definition}[Condition-like number]\label{def:cnum}
    For a zero-sum game $(\gm, -\gm)$ we define 
    \begin{gather*}
    \kappa_\gm = \sup_{(x,y) \not\in \text{NE}}\frac{\min\{\ell_x, \ell_y\}}{\psi(x,y)}
    \end{gather*}
    where for a point $(x,y)$ that is not a Nash equilibrium 
    \begin{gather*}
        \ell_x = \max_i e_i^\top \gm y -  \max_{i\not\in\brset(y)} e_i^\top \gm y \text{ and }\\
        \ell_y = \min_{j\not\in\brset(x)} x^\top \gm e_j - \min _j x^\top \gm e_j.
    \end{gather*}
\end{definition}

Essentially, this constant is a local measure that relates the separation between best and non-best responses and the duality gap. While we were unable to bound this quantity, we note that it cannot go to zero as the duality gap decreases, at least in non-degenerate games. This follows from the strict-complementarity slackness of \cite{goldman1956theory} and indicates that the convergence could be accelerating as we approach the equilibrium, since strict slackness does not apply when far from it.

\begin{theorem}[Last-iterate convergence]\label{th:main}
    For sufficiently small $\delta$, AGFP converges to a $\mathcal{O}\left(\sqrt{\frac{\delta}{\kappa_\gm}}\right)$-Nash equilibrium of a zero-sum game $(\gm, -\gm)$ at a rate of $\mathcal{O}\left(\frac{1}{\kappa_\gm T}\right)$.
\end{theorem}
\begin{proof}
We begin with the rate part of the proof, and relate $\eta_t$ to $\psi(z_t)$. We invoke \Cref{lem:uniquestab} or \Cref{obs:stab} to obtain $\bar{\theta}_x, \bar{\theta}_y$ for the two players, and let $\bar{\theta} = \min\{\bar{\theta}_x, \bar{\theta}_y\}$. Note that the nominators of both $\bar{\theta}_x, \bar{\theta}_y$ can be lower bounded by $\min\{\ell_x, \ell_y\}$, as defined above, and the denominators can be upper bounded by $2$. Hence, it holds that
\[
\eta_t \ge \bar{\theta} \ge \frac{\min\{\ell_x, \ell_y\}}{2} = \frac{\kappa_\gm}{2} \psi(x_t,y_t)
\]
where for the equality we substituted from \Cref{def:cnum}. 

Combining this bound with \Cref{lem:bndinc} we reach 
\begin{align*}
\psi(z_{t+2}) &\le \psi(z_{t}) - \frac{\kappa_\gm}{4}\psi^2(z_t).
\end{align*}
This recurrence is standard and gives us the claimed rate. For the approximation part, note that \Cref{lem:bndinc} breaks when $\delta > \eta_t \psi(z_t) \implies \delta > \kappa_\gm \psi^2(z_t)$, concluding the proof.
\end{proof}

\begin{corollary}[Adaptive $\delta$s] There is an adaptive choice $\delta_t$, i.e. one $\delta$ per iteration, such that the algorithm maintains the same convergence rate asymptotically.
\end{corollary}

\begin{proof}
    The conditions for $\delta_t$ in order for \Cref{th:main} to hold are determined by the proofs of \Cref{lem:tie} and \Cref{lem:bndinc}.  
    All necessary quantities can be determined by the previous iteration, e.g., \Cref{lem:bndinc} imposes an upper bound dependent on $\eta_{t-1}\psi(z_{t-1})$. By setting each $\delta_t$ to respect these bounds, we can guarantee the convergence holds for any horizon T. 
\end{proof}

\section{Experiments}
\label{sec:exp}

Given the rising popularity of Fictitious Play in practical applications, we conclude the main body of our paper with a series of experiments that demonstrate the effectiveness of our approach. 

\paragraph{Stepsize selection and implemention details.} The problem of optimizing the duality gap along the direction specified by a pair of best responses corresponds to minimizing a convex and piecewise linear function. This can be solved exactly in $\mathcal{O}(m + n)$ time, since there are at most $m$ (resp. $n$) changes for the $x$ (resp. $y$) player. Despite being of polynomial time, this would cause a significant time overhead in high dimensional games since the computation of each part requires a matrix-vector multiplication. Instead, we follow an approach similar to that of \cite{DBLP:conf/ijcai/FasoulakisMRS25}, who perform a ternary search since the function is unimodal. We take their approach one step further, and perform a binary search directly on the slopes of the pieces. That method solves up to machine accuracy (roughly $10^{-16}$ for 64bit machines) in about 50 iterations. More importantly, the approach also suits our goal of satisfying $\eta_t \ge \delta$, by setting the termination criterion to accuracy $\delta$. Finally, all the algorithms of this section were developed in Python using the library NumPy and tested on a home computer with an Apple M4 chip and 16GB RAM. 

\paragraph{The effect of $\delta$.} In our original experiment on the Rock-Paper-Scissor game we set $\delta = 10^{-8}$, which can also be deduced by the ``small steps'' in \Cref{fig:agf_rps_steps} (due to binary search, those were $2^{-28}$). To show that the algorithm indeed gets stuck in approximate equilibrium point, we repeat the experiment with $\delta = 10^{-4}$. The results are presented in \Cref{fig:rps_small_d}. Notice how the duality gap plateaus a bit under $10^{-2} = \mathcal{O}(\sqrt{\delta})$. This phenomenon aligns with the statement of \Cref{th:main}, when not accounting for the constant $\kappa_\gm$.

\begin{figure}[H]
    \centering
    \includegraphics[scale=.7]{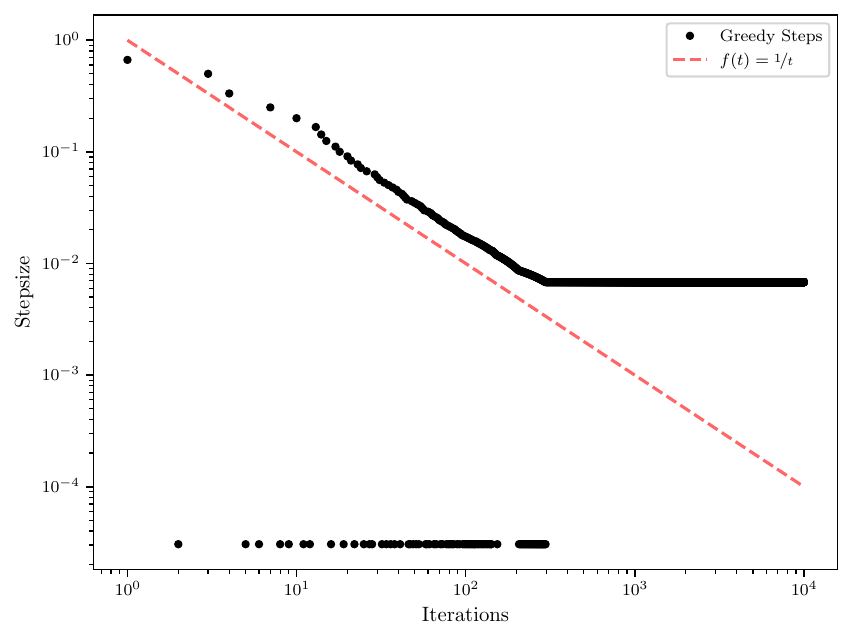} 
    \caption{AGFP in the RPS game with $\delta = 10^{-4}$}
    \label{fig:rps_small_d}
\end{figure}

\paragraph{Performance in random games.} We move forward by testing our method in randomly generated games. More specifically, we sample the entries of the matrix independently from a standard Gaussian distribution and normalize them in $[0,1]$. Since we have already tested the ``small'' Rock-Paper-Scissors game, we first test at a moderate game of size $50 \times 50$ using $\delta = 10^{-8}$. The results are presented in \Cref{fig:random50}. At first sight, it seems that the convergence of $\mathcal{O}\big(\frac{1}{T})$ fails. However, upon closer inspection, we can see that the line gets parallel with the previous reference, until we approach the point of plateauing (around $10^{-4}$). 

\begin{figure}[H]
    \centering
    
    \begin{subfigure}[b]{0.45\textwidth}
        \centering
        \includegraphics[width=\textwidth]{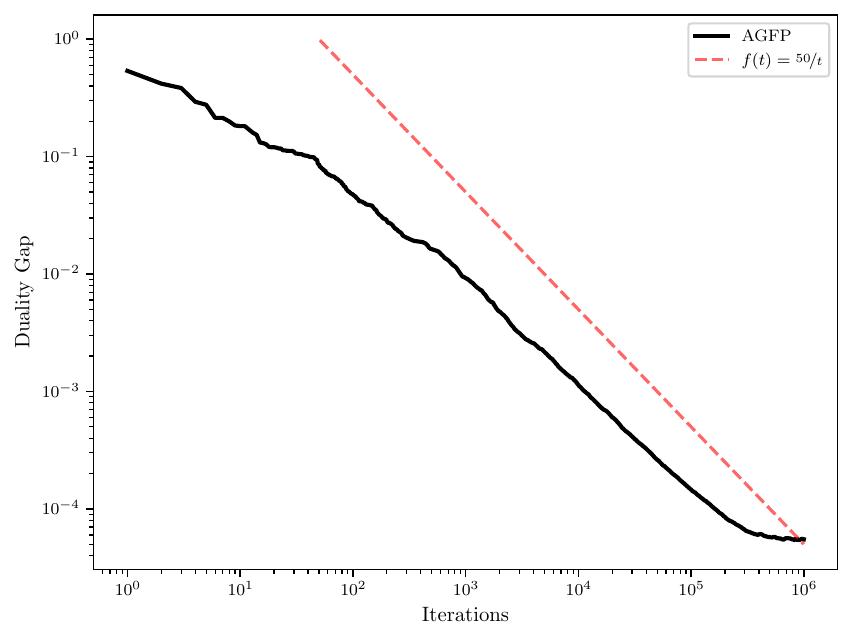}
        \caption{Duality gap per iteration}
    \end{subfigure}
    \hfill
    \begin{subfigure}[b]{0.45\textwidth}
        \centering
        \includegraphics[width=\textwidth]{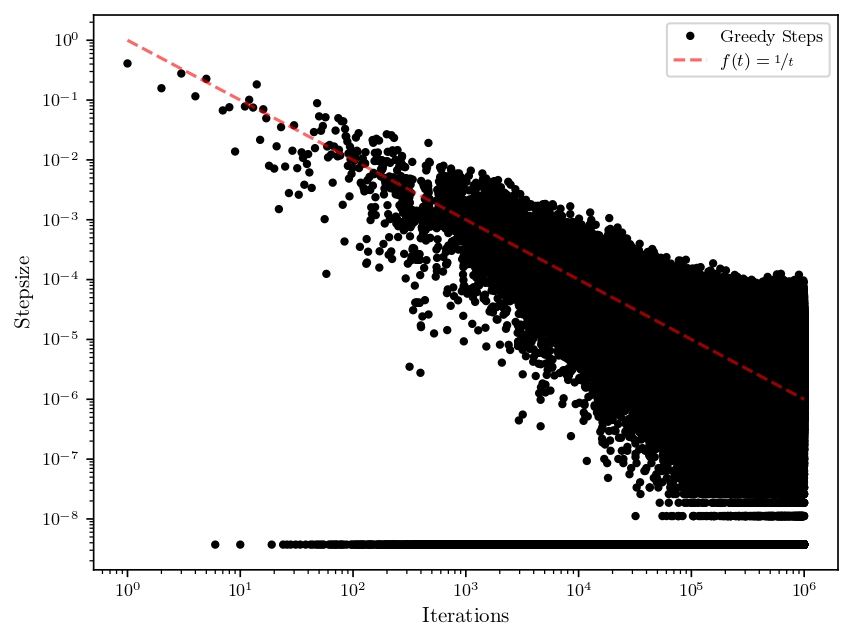}
        \caption{Step per iteration}
    \end{subfigure}
    
    \caption{AGFP in a $50\times 50$ random Gaussian game.}\label{fig:random50}
\end{figure}

Before exploring further the effect of the dimension in the dynamics, we have to comment about the optimal steps in the random case. We observe that the nice picture of the Rock-Paper-Scissors games changes: the stepsizes are still decreasing, but they do not fit the line $\nicefrac{1}{t}$ well. Still, it seems as they are distributed around it.

\paragraph{The effect of the dimension.} In the Rock-Paper-Scissors game, which is $3 \times 3$, we experimentally observed the rate to be $\frac{1}{T}$. Then, in the $50 \times 50$ random Gaussian game we observe $\frac{50}{T}$. We increase the dimension of the experiment by an order of magnitude and test our method on a $500 \times 500$ Gaussian game, setting $\delta = 10^{-11}$ and increasing the number of iterations to get a clearer image. Indeed, in \Cref{fig:random_500} we once again see a rate that scales as $\mathcal{O}\big(\frac{n}{T}\big)$. While all the illustrations presented here are for square matrices, we observed that in the non-square case, the duality gap scales with the maximum of the dimension.
\vspace{-.1cm}
\begin{figure}[H]
    \centering
    \includegraphics[scale=.6]{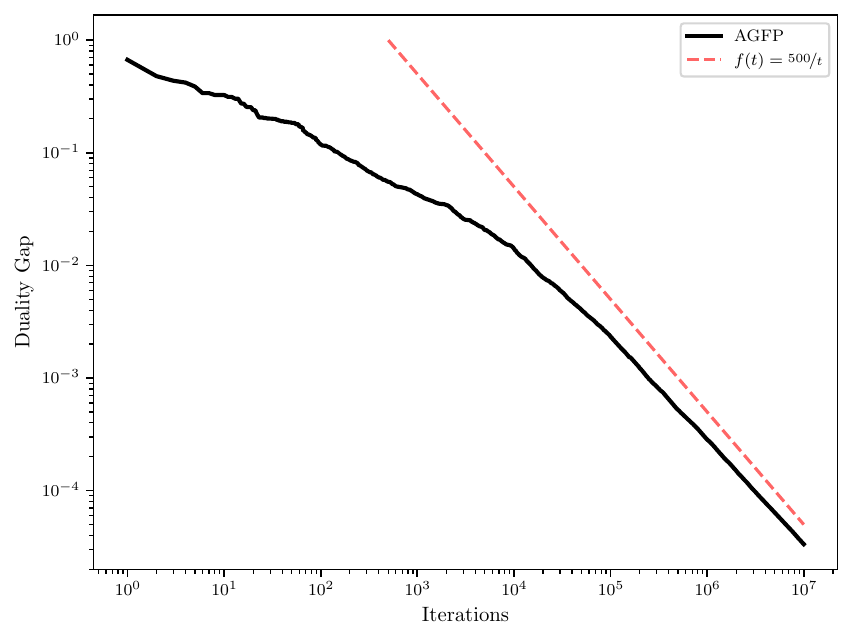} 
    \caption{AGFP in a $500\times 500$ random Gaussian game}
    \label{fig:random_500}
\end{figure}

\paragraph{Wall-clock time.}
Standard Fictitious Play is observed to converge at a rate  $\mathcal{O}\big(\frac{1}{\sqrt{T}}\big)$. Our proposed method improves this rate, but at the expense of increased computational overhead per iteration. To assess whether this trade-off yields a net benefit in practice, we benchmark the wall-clock time-to-accuracy performance of both algorithms. The results are presented in \Cref{fig:time} for a $50 \times 50$ Gaussian game with $\delta = 10^{-10}$. We observe that already at an approximation of $10^{-3}$ our algorithm is faster, while for $10^{-4}$ the difference is striking. The classical algorithm requires around 195 seconds, while our method just over 4\footnote{For an equal base of comparison we run the experiment twice, once to collect the figure's data and one to report the times, to avoid penalizing FP for the increased memory access.}
, a $50\times$ improvement. 

\begin{figure}[H]
    \centering
    \includegraphics[scale=.7]{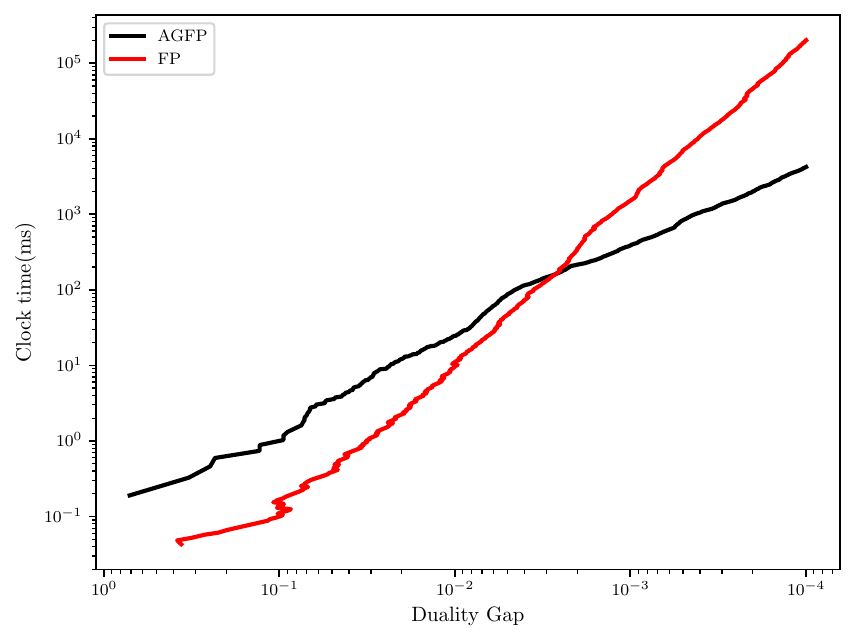} 
    \caption{Time comparison between FP and AGFP in a Gaussian random $50\times 50$ game}
    \label{fig:time}
\end{figure}

\section{Discussion and Future Work}
We began our exploration with the Rock-Paper-Scissors game. It is only natural to conclude with it. We start with the standard $3\times 3$ game. Note that under our symmetric and pure initialization, the only time where the best response was unique was at the first step. At any other time, two of three strategies yielded the same playoff, since all three can only occur at the equilibrium. We conjecture that if the algorithm had memory, then the tie-breaking where it picks a different best response than the previous iteration converges at the rate of $\frac{1}{T}$ with the stepsize sequence $\eta^*_1 = \nicefrac{2}{3}$ and $\eta^*_t = \nicefrac{1}{t},\ t\ge 2$.

Now, the next interesting question is whether that sequence is the optimal one for RPS in higher dimensions. The immediate answer is negative, but there is more to that. In \Cref{fig:rps17} we present the steps the algorithm selects for an RPS game in higher dimension, where we clipped the noisy steps due to $\delta = 10^{-11}$. We still see a decreasing line that is close to $\nicefrac{1}{t}$ but now there are two increasing straight lines as well. We do not have a theoretical justification for this phenomenon and leave it as an open question. Another avenue for future work is to study our algorithm, and especially the condition-like number, in the smooth regime. \cite{DBLP:conf/nips/AnagnostidesS24} provided a similar treatment to gradient-based algorithms. 
Finally, we must note that interesting pursuits are to adapt our idea to non zero-sum games. It seems especially promising to explore the direction of potential games. Since those games admit a pure Nash equilibrium, once the best response direction is this point, our algorithm will jump to it.

Finally, we turn to Karlin's conjecture and our constant $\kappa$. The conjecture is still open for random games. We think that establishing an unconditional convergence of our method for random games via bounding the condition-like number can act as an intermediate step to attack the conjecture, by measuring the duality gap increases of the fixed step schedule to the optimal one.

\begin{figure}
    \centering
    \includegraphics[scale=.7]{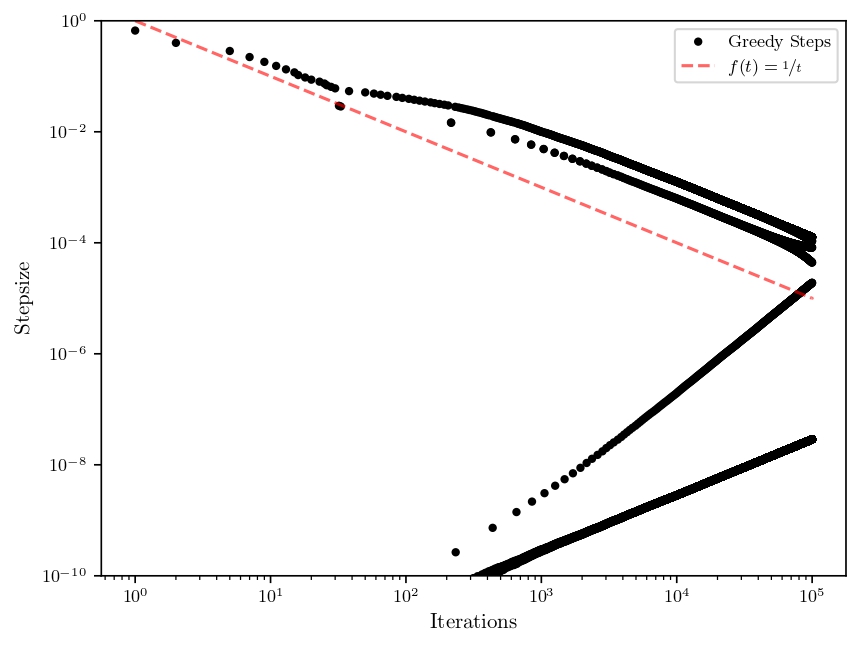}
    \caption{AGFP steps in the $17 \times 17$ RPS game}
    \label{fig:rps17}
\end{figure}

\section*{Acknowledgments}
We would like to thank Michail Fasoulakis, Ioannis Kakatelis and Nikoleta Sevastaki for useful discussions. This work was partially supported by the framework of the H.F.R.I call “Basic research Financing (Horizontal support of all Sciences)” under the National Recovery and Resilience Plan “Greece 2.0” funded by the European Union – NextGenerationEU (H.F.R.I. Project Number: 15877) and by the project MIS 5154714 of the National Recovery and Resilience
Plan “Greece 2.0” funded by the European Union under the NextGenerationEU Program.

\bibliographystyle{plainnat}
\bibliography{references}

\end{document}